\documentclass[acmtog,nonacm]{acmart}
\usepackage{booktabs} % For formal tables
\usepackage{graphicx} % Required for inserting images
\usepackage{amsmath}
\usepackage{amsfonts}
\usepackage{geometry}
\usepackage{verbatim}
\usepackage{bm}
\usepackage{setspace}
\usepackage{amsthm}
\usepackage{wrapfig}
\usepackage{multirow}
\usepackage{makecell}
\usepackage[ruled]{algorithm2e} % For algorithms

\def\RR{\mathbb{R}}

\newcommand{\hide}[1]{}

\def\lambdaLame{\lambda_{\textrm{Lam\'e}}}
\def\muLame{\mu_{\textrm{Lam\'e}}}

\usepackage[ruled]{algorithm2e} % For algorithms

\SetAlFnt{\small}
\SetAlCapFnt{\small}
\SetAlCapNameFnt{\small}
\SetAlCapHSkip{0pt}

\acmJournal{TOG}
\begin{document}
% Title portion
\title{Tuning Nonlinear Elastic Materials under Small and Large Deformations}

% DO NOT ENTER AUTHOR INFORMATION FOR ANONYMOUS TECHNICAL PAPER SUBMISSIONS TO ACM SIGGRAPH

\author{Huanyu Chen}
\affiliation{%
  \institution{University of Southern California}
  \city{Los Angeles}
  \country{USA}}
\email{huanyuc@usc.edu}
\author{Jernej Barbi\v{c}}
\affiliation{%
  \institution{University of Southern California}
  \city{Los Angeles}
  \country{USA}}
\email{jnb@usc.edu}

\begin{abstract}
In computer graphics and engineering, nonlinear elastic
material properties of 3D volumetric solids are typically
adjusted by selecting a material family, such as St. Venant Kirchhoff,
Linear Corotational, (Stable) Neo-Hookean, Ogden, etc., and then
selecting the values of the specific parameters for that family,
such as the Lam\'{e} parameters, Ogden exponents, or whatever the parameterization
of a particular family may be.
However, the relationships between those parameter values,
and visually intuitive material properties such as object's ``stiffness'',
volume preservation, or the ``amount of nonlinearity'', are less clear
and can be tedious to tune.
For an arbitrary isotropic hyperelastic energy density function $\psi$
that is not parameterized in terms of the Lam\'{e} parameters,
it is not even clear what the Lam\'{e} parameters and Young's modulus
and Poisson's ratio are.
Starting from $\psi,$ we first give a concise definition of 
Lam\'{e} parameters, and therefore Young's modulus and Poisson's ratio.
Second, we give a method to adjust the object's three salient properties,
namely two small-deformation properties (overall ``stiffness'', and 
amount of volume preservation, 
prescribed by object's Young's modulus and Poisson's ratio), 
and one large-deformation property 
(material nonlinearity). We do this in a manner whereby each of
these three properties is decoupled from the other two properties,
and can therefore be set independently. This permits a new ability, namely
``normalization'' of materials: starting from two distinct materials,
we can ``normalize'' them so that they have the same small deformation properties,
or the same large-deformation nonlinearity behavior, or both.
Furthermore, our analysis produced a useful theoretical result,
namely it establishes that Linear Corotational materials
(arguably the most widely used materials in computer graphics)
are the simplest possible nonlinear materials.
\end{abstract}




\begin{teaserfigure}
\centerline{\includegraphics[width=0.9\hsize]{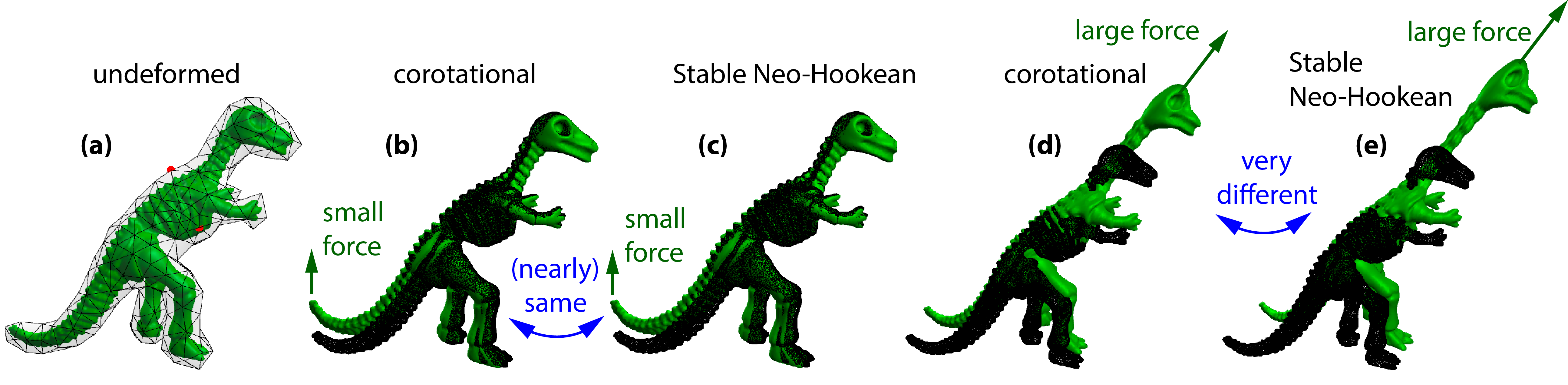}}
\vspace{-0.1cm}
\caption{\textbf{Volumetric material normalization.}
   We compare a corotational solid volumetric material~\cite{Mueller:2004:IVM} against the Stable Neo-Hookean (SNH) material~\cite{Smith:2018:SNF}
   (elastic energy functions are given in Table~\ref{tab:materials}). Red vertices are constrained (a).
   Observe that for SNH, $\lambdaLame = \lambda - \mu,\ \muLame=\mu,$ where $\lambda,\mu$ are the SNH material parameters
   (Table~\ref{tab:materials}), but for corotational material, $\lambdaLame = \lambda,\ \muLame=\mu.$
   If one naively uses SNH and misindentifies the SNH parameter $\lambda$ as $\lambdaLame,$
   the resulting stiffness matrix in the undeformed
   shape differs from that of the corotational material. The frequencies of the resulting linear vibrational modes
   also mismatch; at $\nu=0.2,$ the lowest frequency is off by 18\%. Instead, given any particular desired
   $E$ and $\nu,$ we determine $\lambdaLame$ and $\muLame$ 
   using Equation~\ref{eq:lambdaAndMuLameFormula}.
   We then solve for $\lambda$ and $\mu,$ separately for SNH and corotational, using Table~\ref{tab:materials}.
   The result is that now both materials have the same stiffness matrix (and mass matrix) in the rest shape; this is already predicted by our theory,
   but we also checked it experimentally.
   Effectively, this means that we ``normalized'' SNH to match the corotational material under small deformations.
   Consequently, for small loads, the two materials produce almost the same output (b,c).
   We can then observe the difference between these two materials due to the nonlinearities;
   they are revealed under large forces (d,e).
}
  \label{fig:volumetricMaterialNormalization}
\end{teaserfigure}

\maketitle

\section{Introduction}

Adjusting deformable object material properties to meet some artistic goal
is a commonly encountered operation in computer graphics practice.
There is a plethora of available material model families,
each of which incorporates a small number of material parameters to tune.
When investigating all these materials, it is useful to be able
to ensure that the materials all share the same small deformation properties,
so that the differences in their large-deformation nonlinear properties can be 
investigated systematically against a common small deformation baseline.
By small deformation properties, we mean the material's two Lam\'{e} parameters,
or equivalently (via bijective formulas), Young's modulus and Poisson's ratio.
While many material families do already parameterize the materials
via Lam\'{e} parameters, this is not universally true.
And so the first question that we address is, ``Is it possible
to meaningfully \emph{define} Lam\'{e} parameters (and consequently
Young's modulus and Poisson's ratio) for an arbitrary isotropic
hyperelastic nonlinear material? We provide such a definition
in this paper, and it matches the Lam\'{e} parameters for the common
materials. The key idea is to perform a second-order Taylor
expansion of $\psi$ in terms of $\lambda_i$ 
around the rest shape ($\lambda_i=1,\,i=1,2,3$).
This process is identical to approximating the first Piola-Kirchhoff (PK1)
principal stresses 
(singular values of the PK1 tensor $P=d\psi/dF$; $F$ is the $3\times 3$ deformation gradient)
at the rest shape as a linear function of $\lambda_i,$
a process that we call PK1-linearization. PK1-linearization correctly incorporates
rotations, and as such does not suffer from artifacts under large deformations.
Note that in deformable object modeling, a more common linearization
is the one of internal forces around the rest shape using the stiffness matrix;
however, this linearization is not rotation-invariant and quickly produces
visible artifacts. We prove that for any $\psi,$ its PK1-linearization is
a Linear Corotational material. Obviously, any Linear Corotational material
is a PK1-linearization of at least one material, namely itself. This means
that Linear Corotational materials are ``special'', namely they are
precisely the materials that do not have any third-order or higher terms
in their Taylor expansion in terms of $\lambda_i$ around the rest shape.
This sheds a new light on this widely used material: 
in light of our result, the popularity of Linear Corotational materials is not surprising; 
they are arguably the simplest possible nonlinear isotropic materials.
These observations also lead to a natural definition
of Lam\'{e} parameters for an arbitrary $\psi.$ Namely, compute the Linear
Corotational material that is its PK1-linearization, and read the
Lam\'{e} parameters of that material.
These observations allow us to find a material in any material family
(with at least two parameters) that matches a given Young's modulus and Poisson's ratio.
They also allow us to ``normalize'' two nonlinear materials
so that they have the same quadratic Taylor expansion (or equivalently, PK1-linearization),
i.e., they share the same Lam\'{e} parameters, or equivalently Young's modulus
or Poisson's ratio.

With the small-deformation regime covered, we then propose a method
to modify any material $\psi$ so that (1) its PK1-linearization remains the same,
and (2) the material progressively stiffens or softens under large deformations in a prescribed
manner, using a single scalar parameter that we call ``nonlinearity''.
This makes it easily possible to make a material more or less nonlinear.
An existing alternative is to change the material family for one that is inherently more
(or less) nonlinear (e.g., StVK is known to be more nonlinear than the Stable Neo-Hookean material).
However, this still limits the choice to a few discrete material families; whereas with
our method, we obtain a continuously-varying nonlinearity adjustment.

%%%%%%%%%%%%%%%%%%%%%%%%%%%%%%%%%%%%%%%%%%%%%%%%%%%%%%%%%%%%%%%%%

\section{Related Work}
\label{sec:relatedWork}

Our work is applicable generally to 3D solid deformable object simulation;
it acts as a mechanism to adjust the material properties to meet specific
user requirements, and is agnostic of the specific underlying simulator.
Nonlinear deformable object simulation has a long history in computer graphics;
for good surveys, please see~\cite{Sifakis:2012:FSO,Kim:2020:DD}.
We use ``stretch-based'' materials~\cite{Xu:2015:NMD,Chen:2023:CAR} to
express the elastic energy density function. These are materials where
the elastic energy density function is an isotropic function of the three principal
stretches; we use them because they are general, and recent work demonstrated how
to robustly remove singularities~\cite{Chen:2023:CAR}.
Designing 3D solid material properties to meet specific artistic goals
has been investigated in~\cite{Wang:2020:ACS}; however, the paper
only discussed adjusting ``stiffness'' (via Young's modulus);
and did not discuss Poisson's ratio, Lam\'{e} parameters or nonlinearity
adjustments. For plants~\cite{Wang:2017:BMB}, it has been shown
how to adjust model-reduced models to match desired user frequencies.
There has not been a prior method that has demonstrated how to 
adjust nonlinearity and small-deformation behavior with only three
compact continuous parameters. The definitions of ``PK1-linearity''
and the observations and proofs on the special role of the Linear Corotational
material are also novel. Linear Corotational material has been introduced
to computer graphics in~\cite{Mueller:2004:IVM}, and is widely used,
e.g.,~\cite{Chao:2010:ASG,McAdams:2011:EEF,Hecht:2012:USC,Verschoor:2018:SHS,Arriola:2020:MOD}.

There are many nonlinear material models used in computer graphics
and engineering. Our method is generic and works with an arbitrary
isotropic material. We explicitly provide our defined and computed
Lam\'{e} parameters (Table~\ref{tab:materials}) for many commonly used materials.
These include materials from the Seth-Hill family~\cite{Seth:1964:GSM}
(Linear Corotational, StVK, Hencky, Symmetric Seth-Hill~\cite{Bazant:1998:ETC},
Hill family~\cite{Hill:1968:OCI}); Neo-Hookean materials
(standard Neo-Hookean, Stable Neo-Hookean~\cite{Smith:2018:SNF}, 
STS~\cite{Pai:2018:THT}); Valanis-Landel materials~\cite{Valanis:1967:TSE,Valanis:2022:TVL,Xu:2015:NMD,Peng:1972:SEF}; materials used in geometric modeling 
(ARAP~\cite{Sorkine:2007:ARA}, Symmetric ARAP~\cite{Shtengel:2017:GOV},
Symmetric Dirichlet~\cite{Smith:2015:BPW}); and Ogden~\cite{Ogden:1972:LDI}
and Mooney-Rivlin materials. For a comprehensive review of isotropic hyperelastic material models, please refer to~\cite{Beda:2014:AAF,Melly:2021:ARO,He:2022:ACS}. 

%%%%%%%%%%%%%%%%%%%%%%%%%%%%%%%%%%%%%%%%%%%%%%%%%%%%%%%%%%%%%%%%%

\section{PK1-Linearization of Isotropic Materials}
\label{app:volElastic}

What happens if we approximate a 3D isotropic 
hyperelastic material elastic energy density function $\psi(\lambda_1,\lambda_2,\lambda_3)$
up to the second order around the rest shape in terms of $\lambda_1,\lambda_2,\lambda_3$?
This question is important because such an approximation controls 
the small-deformation behavior, while still resulting in a nonlinear simulation
that does not exhibit artifacts under large deformations, as explained below.
Equivalently, this can be seen as a linearization of the first Piola-Kirchhoff (PK1) principal stresses
$p_i = \partial \psi / \partial\lambda_i\in\RR$ around the rest shape in terms $\lambda_i.$
Recall that the PK1 stress tensor $P\in\RR^{3\times 3}$ is computed as~\cite{Irving:2004:IFE,Teran:2005:RQF}
\begin{gather}
P = U 
\textrm{diag}\Bigl(
  p_1(\lambda_1,\lambda_2,\lambda_3),
  p_2(\lambda_1,\lambda_2,\lambda_3),
  p_3(\lambda_1,\lambda_2,\lambda_3)
\Bigr)
V^T,\\ 
\textrm{for\ }\quad
F = U
\textrm{diag}\Bigl(\lambda_1, \lambda_2,\lambda_3\Bigr)
V^T,
\end{gather}
where $F\in\RR^{3\times 3}$ is the deformation gradient, $U,V\in\RR^{3\times 3}$ are orthogonal matrices
from the SVD of $F,$ $\lambda_i\in\RR$ are the singular values of $F$ (the ``principal stretches''),
and $p_i = \partial \psi / \partial\lambda_i\in\RR$ are the \emph{principal} PK1 stresses, i.e.,
singular values of $P.$ Because of the presence of $U$ and $V,$ if we linearize the PK1
principal stresses in terms of $\lambda_i$ (or equivalently, approximate $\psi$ up to second
order in terms of $\lambda_i$), the simulation remains rotation-aware
and will not exhibit geometric linearity artefacts.
The PK1-linearization is governed by two 
linear parameters (Young's modulus $E$, Poisson's ratio $\nu$) 
or equivalently, $(\lambdaLame,\muLame),$ as established by the following theorem.
\begin{theorem} 
\label{the:volumetric2nd}
   The 2\textsuperscript{nd}-order Taylor expansion of any 3D isotropic hyperelastic volumetric material $\psi$ as a function of principal stretches $\lambda_1, \lambda_2, \lambda_3$ around the rest shape $\lambda_1 = \lambda_2 = \lambda_3 = 1$ is a Linear Corotational material. 
The Lam\'{e} parameters of this Linear Corotational material are a unique function of $\psi,$ as follows:
\begin{gather}
\label{eq:defLambda}
      \lambdaLame = \partial_{12} \psi(1, 1, 1),\\
\label{eq:defMu}
      \muLame = \frac{1}{2} \bigl(\partial_{11} \psi(1, 1, 1) - \partial_{12} \psi(1, 1, 1)\bigr).
\end{gather}
\end{theorem}
\begin{proof}
   The $2^{\textrm{nd}}$-order Taylor expansion of $\psi(\lambda_1, \lambda_2, \lambda_3)$ around $(1, 1, 1)$ is
   \begin{equation}
      \psi(1, 1, 1) + \sum_{i=1}^3 \partial_i \psi(1, 1, 1) (\lambda_i - 1)
      + \frac{1}{2} \sum_{i=1}^3 \partial_{ij}\psi(1, 1, 1) (\lambda_i - 1)(\lambda_j - 1),
   \end{equation}
   where
   \begin{equation}
      \partial_i \psi = \frac{\partial \psi}{\partial \lambda_i}, \qquad
      \partial_{ij} \psi = \frac{\partial^2 \psi}{\partial \lambda_i \partial \lambda_j}.
   \end{equation}
   Because $\psi(1, 1, 1) = 0$ and $\partial_i \psi(1, 1, 1) = 0,$ for $i = 1, 2, 3$, the 0\textsuperscript{th}-order and the 1\textsuperscript{st}-order terms disappear. The 2\textsuperscript{nd}-order term can be further simplified since $\psi$ is a symmetric function of $\lambda_1, \lambda_2, \lambda_3$,
   \begin{gather}
      \partial_{11} \psi(1, 1, 1) = \partial_{22} \psi(1, 1, 1) 
      = \partial_{33} \psi(1, 1, 1), \\
      \partial_{12} \psi(1, 1, 1) = \partial_{13} \psi(1, 1, 1) 
      = \partial_{23} \psi(1, 1, 1),
   \end{gather}
It follows that
   \begin{gather}
      \frac{1}{2} \sum_{i=1}^3 \partial_{ij}\psi(1, 1, 1) (\lambda_i - 1)(\lambda_j - 1)=
      \\\nonumber
      = \frac{1}{2} \partial_{11} \psi(1, 1, 1) 
      \bigl((\lambda_1 - 1)^2 + (\lambda_2 - 1)^2 + (\lambda_3 - 1)^2\bigr)+
      \\\nonumber
      + \partial_{12} \psi(1, 1, 1) 
      \bigl((\lambda_1 - 1)(\lambda_2 - 1) + (\lambda_1 - 1)(\lambda_3 - 1)+\\
      + (\lambda_2 - 1)(\lambda_3 - 1)\bigr)=
      \\\nonumber
      = \frac{1}{2} \bigl(\partial_{11} \psi(1, 1, 1) - \partial_{12} \psi(1, 1, 1)\bigr)
      \bigl((\lambda_1 - 1)^2 + (\lambda_2 - 1)^2 +\bigr.\\
                \bigl. + (\lambda_3 - 1)^2\bigr)
      + \frac{1}{2} \partial_{12} \psi(1, 1, 1) 
      (\lambda_1 + \lambda_2 + \lambda_3 - 3)^2.
   \end{gather}
   But, this is exactly the Linear Corotational material with
   \begin{gather}
      \lambdaLame = \partial_{12} \psi(1, 1, 1),\\
      \muLame = \frac{1}{2} (\partial_{11} \psi(1, 1, 1) 
      - \partial_{12} \psi(1, 1, 1)).
   \end{gather}
\end{proof}
The theorem implies that the behavior of any volumetric solid isotropic material around 
the rest shape is determined by 
two quantities $\partial_{11} \psi (1, 1, 1)$ and $\partial_{12} \psi (1, 1, 1)$, 
and they provide the \emph{quadratic expansion} of the material. 
Therefore, this enables us to \emph{define} the Lam\'e parameters of an 
arbitrary volumetric isotropic material $\psi(\lambda_1, \lambda_2, \lambda_3)$ 
using the formulas in Equation~\ref{eq:defLambda} and~\ref{eq:defMu}!
Note that several (but not all) popular material families are already defined in terms 
of Lam\'e parameters; for those families our definition above matches those Lam\'e parameters.
In Section~\ref{sec:normalization}, we give a comprehensive table listing the two Lam\'e parameters 
for many commonly used materials.
Even though two volumetric materials may be distinct and even
exhibit very different behavior under large deformations, their PK1-linearization 
behavior is characterized by the above-defined $\muLame$ and $\lambdaLame.$
This observation enables us to ``normalize'' nonlinear volumetric elastic materials so that
their behavior around the rest shape is the same.
Therefore, their differences under large deformations can be more
meaningfully compared. We show an example of such a process in Figure~\ref{fig:volumetricMaterialNormalization}.

\begin{figure}[!t]
\centering
\includegraphics[width=1.0\hsize]{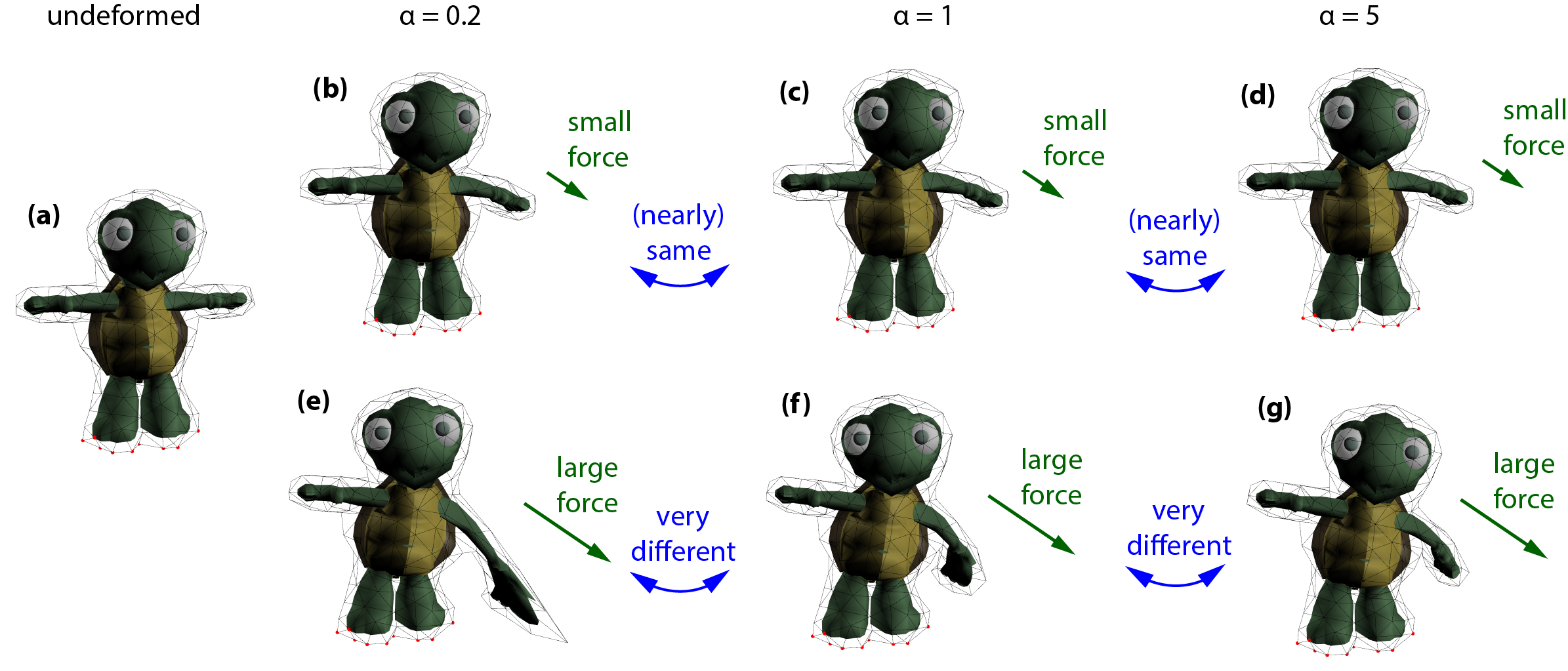}
\caption{
  {\bf By adjusting $\alpha,$ the deformable object can be made more or less stiff under large deformations.}
  The small deformation behavior is unchanged when changing $\alpha.$
} 
\label{fig:nonlinearity}
\end{figure}

\section{Normalization of 3D Volumetric Materials}
\label{sec:normalization}

\begin{table*}[t]
  \centering
  \caption{\textbf{Common 3D volumetric elastic materials and their linear parameters.}
    We have  $J=\lambda_1\lambda_2\lambda_3.$
  }
  \footnotesize
\begin{tabular}{l|c|c|c}\toprule
  Name & $\Psi(\lambda_1,\lambda_2,\lambda_3)$ & $\lambdaLame$ & $\muLame$ \\
  \midrule
Linear Corotational
(S.H. $\mathbf{\alpha = 1}$) &
$\mu \bigl(\sum_{i=1}^3 ( \lambda_{i} - 1 )^{2}\bigr) + \frac{\lambda}{2} (-3 + \sum_{i=1}^3 \lambda_{i})^{2}$
& $\lambda$  & $\mu$ \\[2mm]
St. Venant-Kirchhoff
(S.H. $\mathbf{\alpha = 2}$) &
$\frac{\mu}{4} \bigl(\sum_{i=1}^3 ( \lambda_{i}^{2} - 1 )^{2}\bigr) + 
\frac{\lambda}{8}( -3 + \sum_{i=1}^3 \lambda_{i}^{2} )^{2}$ 
& $\lambda$  & $\mu$ \\[2mm]
Hencky
(S.H. $\mathbf{\alpha = 0}$) &
$\mu \bigl(\sum_{i=1}^3 \log^2\lambda_i\bigr) + \frac{\lambda}2\log^2 J$ 
& $\lambda$  & $\mu$ \\[2mm]

Seth-Hill family (S.H.) \cite{Seth:1964:GSM} &
$\frac{\mu}{\alpha^2}\bigl(\sum_{i=1}^3 (\lambda_i^{\alpha} - 1)^2\bigr) 
+ \frac{\lambda}{2\alpha^{2}}( -3 + \sum_{i=1}^3 \lambda_i^{\alpha} )^2$ 
& $\lambda$  & $\mu$ \\[2mm]

Symmetric Seth-Hill \cite{Bazant:1998:ETC} &
$\frac{\mu}{4\alpha^2}\bigl(\sum_{i=1}^3 
(\lambda_i^{\alpha} - \frac{1}{\lambda_i^{\alpha}})^2\bigr) 
+ \frac{\lambda}{8\alpha^{2}}( \sum_{i=1}^3 
(\lambda_i^{\alpha} - \frac{1}{\lambda_i^{\alpha}}) )^2$ 
& $\lambda$  & $\mu$ \\[2mm]

Hill family \cite{Hill:1968:OCI} &
$\mu \bigl(\sum_{i=1}^3 f^2(\lambda_i)\bigr) 
+ \frac{\lambda}{2}(\sum_{i=1}^3 f(\lambda_i))^2$ 
& $\lambda$  & $\mu$ \\[2mm]

\midrule

Neo-Hookean (standard version) &
$\frac{\mu}{2}\bigl(-3 + \sum_{i=1}^3 \lambda_i^2\bigr) - \mu\log J + \frac{\lambda}{2}\log^2 J$ 
& $\lambda$  & $\mu$ \\[2mm]

Neo-Hookean (Ogden) &
$\frac{\mu}{2}\bigl(-3 + \sum_{i=1}^3 \lambda_i^2 \bigr) - \mu\log J + \frac{\lambda}{2}(J - 1)^2$ 
& $\lambda$  & $\mu$ \\[2mm]

\makecell[l]{Stable Neo-Hookean\\ (Eq. 13 in~\cite{Smith:2018:SNF})} &
$\frac{\mu}{2}\bigl(-3 + \sum_{i=1}^3 \lambda_i^2 \bigr) - \mu(J-1) + \frac{\lambda}{2}(J - 1)^2$ 
& $\lambda - \mu$  & $\mu$ \\[2mm]

STS material \cite{Pai:2018:THT} &
$\frac{\mu}{2}\bigl(-3 + \sum_{i=1}^3 \lambda_i^2\bigr) - \mu\log J + \frac{\lambda}{2}\log^2 J
+ \frac{\mu_4}{8}(\sum_{i=1}^3 (\lambda_i^2 - 1)^4)$
& $\lambda$  & $\mu$ \\[2mm]

\midrule

Valanis-Landel (original) \cite{Valanis:1967:TSE} &
$2\mu \sum_{i=1}^3 \lambda_i(\log\lambda_i - 1) + \frac{\lambda}{2} \log^2 J$
& $\lambda$  & $\mu$ \\[2mm]

Valanis-Landel (new) \cite{Valanis:2022:TVL,Chen:2023:CAR} &
$\sum_{i=1}^3 f(\lambda_i) + h(J)$
& $h''(1)$  & $\frac{1}{2} f''(1)$ \\[2mm]

Valanis-Landel (Xu's version) \cite{Xu:2015:NMD} &
$\sum_{i=1}^3 f(\lambda_i) + \sum_{i,j=1}^3 g(\lambda_i\lambda_j) + h(J)$
& $g''(1) + h''(1)$  & $\frac{1}{2} (f''(1) + g''(1))$ \\[2mm]

Peng-Landel \cite{Peng:1972:SEF} &
$E \sum_{i=1}^3 \left( \lambda_i - 1 - \log \lambda_i
- \frac{1}{6} \log^2 \lambda_i + \frac{1}{18} \log^3 \lambda_i
- \frac{1}{216} \log^4 \lambda_i \right)$
& 0 & $\frac{E}{3}$ \\[2mm]

\midrule

\makecell[l]{ARAP (As-Rigid-As-Possible)\\ \cite{Sorkine:2007:ARA}} &
$||F-R||^2 = \sum_{i=1}^3 (\lambda_i - 1)^2 $ & 0 & 1 \\[2mm]

Symmetric ARAP~\cite{Shtengel:2017:GOV} & 
$\frac{\mu}{2}(||F-R||^2 + ||F^{-1}-R^{-1}||^2) = \frac{\mu}{2} \sum_{i=1}^3
((\lambda_i - 1)^2 + (1 - \frac{1}{\lambda_i})^2) $  & 0 & $\mu$ \\[2mm]

\makecell[l]{Symmetric Dirichlet\\ \cite{Smith:2015:BPW}}  & 
$\frac{1}{2}(||F||^2 + ||F^{-1}||^2) = \frac{1}{2} \sum_{i=1}^3
(\lambda_i - \frac{1}{\lambda_i})^2$ & 0 & 2 \\[2mm]

\midrule

Ogden \cite{Ogden:1972:LDI} &
$\sum_{p=1}^N{\frac{\mu_p}{\alpha_p}\bigl(-3 + \sum_{i=1}^3 \lambda_i^{\alpha_p}\bigr)}$
& 0 & $\frac{1}{2} \sum_{p=1}^N \mu_p (\alpha_p-1)$ \\[2mm]

Mooney-Rivlin &
$C_1 J^{\frac{-2}{3}}\bigl(-3 + \sum_{i=1}^3 \lambda_i^2\bigr)
+ C_2 J^{\frac{-4}{3}}\bigl(-3 + \sum_{i=1}^3 \lambda_i^2\lambda_{1 + (i\, \textrm{mod}\, 3)}^2\bigr)$ 
& $-\frac{4}{3}( 2C_1 + 5C_2 )$  & $C_1$ \\

\bottomrule
\end{tabular}
\label{tab:materials}
\end{table*}

Table~\ref{tab:materials} gives the PK1-linearizations for several common 3D volumetric materials.
For volumetric solids, we can compute Young's modulus $E$ and Poisson's ratio $\nu$ as 
\begin{equation}
E=\frac{\muLame(3\lambdaLame+2\muLame)}{\lambdaLame+\muLame},\quad
\nu=\frac{\lambdaLame}{2(\lambdaLame+\muLame)}.
\end{equation}
Therefore, Ogden, ARAP, Symmetric ARAP and Symmetric Dirichlet have $\nu=0$ (i.e.,
there is no volume preservation); for these materials, we have $E=2\muLame.$
Often in computer graphics and engineering, we want to prescribe the material ``stiffness'' (via $E$),
and volume preservation (via $\nu$).
This can be done using the formulas
\begin{gather}
\label{eq:lambdaAndMuLameFormula}
  \lambdaLame = \frac{E\nu}{(1+\nu)(1-2\nu)},\quad \muLame = \frac{E}{2(1+\nu)}. 
\end{gather}
The next step is then to set the parameters of the specific material (Table~\ref{tab:materials}) to
meet this $\lambdaLame$ and $\muLame.$ For two-dimensional material families this will
lead to a unique solution of the material parameters, as a function of $\lambdaLame$ and $\muLame.$
For three or higher-dimensional families, the solution is not unique, but can be biased
to prefer modifying only any two particular parameters, or modify all parameters by the least amount,
as two possible example approaches.

\begin{figure}[!t]
\centering
\includegraphics[width=0.95\hsize]{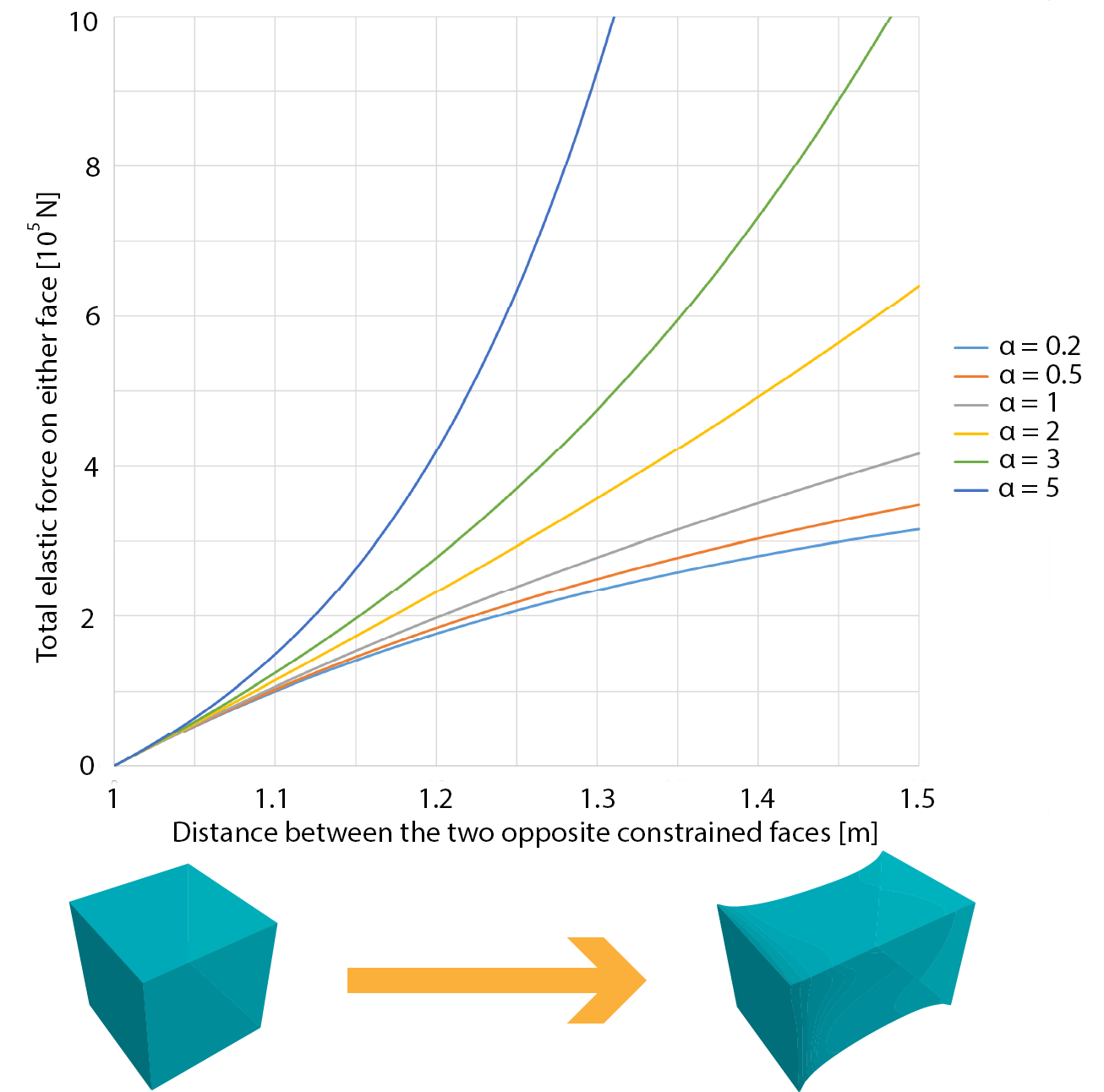}
\caption{
  {\bf Modifying the nonlinearity of the Stable Neo-Hookean (SNH) material.} 
  An elastic cube ($1\textrm{m}\times 1\textrm{m}\times 1\textrm{m}$) is constrained on two opposite faces and pulled apart. We plot the
  total elastic force on either face versus the distance between the two faces. Note that
  the distance is 1.0 when the cube is undeformed. The SNH material ($\alpha=1$)
  is known to be ``soft'' under large deformations; observe that the curve is concave in this case.
  By adjusting $\alpha,$ SNH can be transformed into a much stiffer or softer material.
  Observe that all materials are the same in the linear region; this is thanks
  to using Equation~\ref{eq:nonlinearity} that preserves linear properties.
} 
\label{fig:cube}
\end{figure}

\section{Adjusting the nonlinearity}
\label{sec:nonlinearity}

Given any isotropic function $\psi$ that has been processed as explained
above (Section~\ref{sec:normalization}) to obtain desired linear properties,
we now show how to modify its nonlinearity, without changing the linear properties.
We do this by modifying $\psi(\lambda_1,\lambda_2,\lambda_3)$ to the function
\begin{equation}
\label{eq:nonlinearity}
\psi_\alpha(\lambda_1,\lambda_2,\lambda_3) =
\frac{1}{\alpha^2} \psi(\lambda_1^\alpha,\lambda_2^\alpha,\lambda_3^\alpha),
\end{equation}
for any $\alpha>0.$ 
It is worth noting that the Seth-Hill material family~\cite{Seth:1964:GSM} has a parameter that controls nonlinearity. Our nonlinearization has been inspired by this family, and could be seen as its natural generalization to an arbitrary hyperelastic isotropic material. As a matter of fact, the Seth-Hill family is a special case of our method, namely when our method is applied to the Linear Corotational material. As an important example, the StVK material is obtained from the Linear Corotational material
for $\alpha=2.$ This explains why film VFX practitioners have
observed that the StVK material is ``stiffer'' (more nonlinear; it
progressively stretches less and less when force is increased)
under large deformations than the Linear Corotational material,
and is as such more suitable for 
modeling biological tissues~\cite{Barbic:2024:PC}.

The gradient and Hessian of $\psi_\alpha$ can be easily computed
\begin{gather}
	\label{eq:psiAlphaGradient}
	\nabla \psi_\alpha (\lambda_1,\lambda_2,\lambda_3) =
	\frac{1}{\alpha} \begin{pmatrix}
		\lambda_1^{\alpha - 1} \partial_1 \psi(\lambda_1^\alpha,\lambda_2^\alpha,\lambda_3^\alpha) \\
		\lambda_2^{\alpha - 1} \partial_2 \psi(\lambda_1^\alpha,\lambda_2^\alpha,\lambda_3^\alpha) \\
		\lambda_3^{\alpha - 1} \partial_3 \psi(\lambda_1^\alpha,\lambda_2^\alpha,\lambda_3^\alpha)
	\end{pmatrix}, \\
	\nonumber
	\nabla^2 \psi_\alpha (\lambda_1,\lambda_2,\lambda_3) =
	\frac{\alpha - 1}{\alpha} \textrm{diag}(
	\lambda_1^{\alpha - 2} \partial_1 \psi, \ 
	\lambda_2^{\alpha - 2} \partial_2 \psi, \ 
	\lambda_3^{\alpha - 2} \partial_3 \psi
	) + \\
	\label{eq:psiAlphaHessian}
	+ \begin{pmatrix}
		\lambda_1^{2\alpha - 2} \partial_{11} \psi &
		\lambda_1^{\alpha - 1} \lambda_2^{\alpha - 1} \partial_{12} \psi &
		\lambda_1^{\alpha - 1} \lambda_3^{\alpha - 1} \partial_{13} \psi \\
		\lambda_1^{\alpha - 1} \lambda_2^{\alpha - 1} \partial_{12} \psi &
		\lambda_2^{2\alpha - 2} \partial_{22} \psi &
		\lambda_2^{\alpha - 1} \lambda_3^{\alpha - 1} \partial_{23} \psi \\
		\lambda_1^{\alpha - 1} \lambda_3^{\alpha - 1} \partial_{13} \psi &
		\lambda_2^{\alpha - 1} \lambda_3^{\alpha - 1} \partial_{23} \psi &
		\lambda_3^{2\alpha - 2} \partial_{33} \psi
	\end{pmatrix}.
\end{gather}
Plugging the rest shape $\lambda_1 = \lambda_2 = \lambda_3 = 1$ into Equation~\ref{eq:psiAlphaHessian}, one can easily see that $\nabla^2\psi_\alpha(1, 1, 1) = \nabla^2\psi(1, 1, 1)$. Therefore, the PK1-linearization remains the same. Observe that for $\alpha<1,\alpha=1,\alpha>1$ the material softens, remains the same, and stiffens under large deformations, respectively. Figures~\ref{fig:nonlinearity} and~\ref{fig:cube} demonstrate the
effect of tweaking $\alpha.$
Note that applying the nonlinearization is very easy and requires minimal coding change; 
all is needed is to add the $\alpha$ parameter to the existing code and minimally modify the classes that compute the elastic energy, its gradient and Hessian, 
as given in Equation~\ref{eq:nonlinearity},~\ref{eq:psiAlphaGradient}, and~\ref{eq:psiAlphaHessian}.

\section{Mixing and Matching Materials}
\label{sec:mixAndMatch}

For many materials, the elastic energy $\psi$ is already separated into
the ``$\lambdaLame$'' and ``$\muLame$'' parts (e.g. Seth-Hill family,
the Neo-Hookean materials, STS material, 
and the original and ``new'' Valanis-Landel material; see Table~\ref{tab:materials}),
\begin{equation}
\psi(\lambda_1,\lambda_2,\lambda_3) = 
\lambdaLame \psi_{\lambdaLame}(\lambda_1,\lambda_2,\lambda_3) + 
\muLame \psi_{\muLame}(\lambda_1,\lambda_2,\lambda_3),
\end{equation}
where the $\lambdaLame$ and $\muLame$ of $\psi_{\lambdaLame}$ (as defined in Equation~\ref{eq:defLambda}) are $1$ and $0,$ respectively,
and the $\lambdaLame$ and $\muLame$ of $\psi_{\muLame}$ (as defined in Equation~\ref{eq:defMu}) are $0$ and $1,$ respectively.
This makes it possible to produce new materials by linearly combining $\psi_{\lambdaLame}$ and $\psi_{\muLame}$
from two distinct materials, producing new materials. These materials are sometimes superior to the original materials.
For example, $\psi_{\lambdaLame}$ of Linear Corotational and StVK materials
only provide volume preservation (at the selected $\nu$) for small deformations.
Under large deformation, those formulas no longer corresponds to any kind of volume preservation.
To address this, one can combine $\psi_{\muLame}$ of either of those materials with the $\psi_{\lambdaLame}$
of a Neo-Hookean material (essentially either $(J-1)^2$ or $\log^2(J)$). This produces a material
that obeys the Poisson's effect of the given $\nu$ under small deformations, but still preserve volume under large deformations.
Although such materials have been proposed before
(Linear Corotational with $(J-1)^2$~\cite{Stomakhin:2012:ECI,Smith:2018:SNF}; StVK with $(J-1)^2$~\cite{Zheng:2022:SOH})
our paper makes it possible to discover such combinations systematically.
Also, our approach makes it possible to augment any energy that did not consider volume preservation.
For example, we can augment geometry processing energies with volume preservation, 
giving them the ability to use a non-zero $\nu.$
Another example are incompressible materials in engineering literature. Commonly, the process of creating such materials
works as follows. First, define an energy function $\psi$
\emph{without} any consideration for incompressibility (i.e., $\lambdaLame=\nu=0$). Next, 
enforce exact incompressibility at the solver level, 
using constraints~\cite{Valanis:1967:TSE,Sussman:2009:AMI}. We can take that same energy function $\psi,$
and add a volume-preserving term $\psi_{\lambdaLame}$ of a Neo-Hookean material, with $\lambdaLame$ tuned
using our Equation~\ref{eq:defLambda}. 
This produces a compressible version of the same material that obeys the Poisson's effect of the given $\nu.$

Similarly, we can apply the nonlinearity idea separately to each part, using two independently tweakable parameters
$\alpha_1$ and $\alpha_2,$
\begin{gather}
\label{eq:nonlinearitySeparated}
\nonumber
\psi_{\alpha_1, \alpha_2} (\lambda_1,\lambda_2,\lambda_3) =
\frac{\lambdaLame}{\alpha_1^2} \psi_{\lambdaLame}(\lambda_1^{\alpha_1},\lambda_2^{\alpha_1},\lambda_3^{\alpha_1}) +\\
+ \frac{\muLame}{\alpha_2^2} \psi_{\muLame}(\lambda_1^{\alpha_2},\lambda_2^{\alpha_2},\lambda_3^{\alpha_2}). 
\end{gather}

\section{Results}

We implemented our results on a 3.00 GHz Intel Xeon i7 
CPU E5-2687W v4 processor with 48 cores.
The effect of applying the nonlinearization to computation times is negligible.
Our results enable the artist to adjust both the small-strain (rotation-aware) behavior
as well as nonlinear behavior under large deformations.
The dinosaur example (Figure~\ref{fig:volumetricMaterialNormalization})
illustrates our contribution of adjusting the PK1-linearization
behavior of elastic 3D solids. In this manner, two materials can be
made to have equal behavior under small deformations,
enabling one to more easily compare
their nonlinearities under large deformations.
Figures~\ref{fig:nonlinearity} and~\ref{fig:cube} demonstrate the effect of adjusting nonlinearity.

%%%%%%%%%%%%%%%%%%%%%%%%%%%%%%%%%%%%%%%%%%%%%%%%%%%%%%%%%%%%%%%%%

\section{Conclusion}

While the Linear Corotational material is widely used, we 
provided a formal framework and proved in it that this material is ``special''
and arguably the simplest kind of an isotropic material.
We provided an algorithm to approximate any isotropic material
with its ``PK1-linearization'', namely the Linear Corotational material
that shares the same Lam\'{e} parameters, or equivalently, the same
Young's modulus and Poisson's ratio.
We analytically derived the PK1-linearizations for many common materials (Table~\ref{tab:materials}),
and gave a method to minimally modify any material so that it has a desired
target PK1-linearization. Finally, we demonstrated how to easily adjust the material nonlinearity
using an intuitive one-dimensional parameter family, while keeping PK1-linearization
constant. This makes it possible for artists to easily adjust both the small-deformation
``stiffness'' and volume preservation of a material, as well as adjust the rate at which
the material stiffens or softens under large deformations. Arguably, such a 
$3$-dimensional family provides very good modeling power for many digital artists,
while retaining simplicity of having a very small number of parameters.
Of course, nonlinearity is much more than just tweaking one parameter,
and we did not investigate how nonlinearity could be adjusted by using several (more than two) meaningful parameters.
Our volume preservation is controled via the Poisson's ratio $\nu,$ and as such the Poisson's effect is only
accurate under small deformations. Under large deformations Poisson's ratio 
loses its meaning; this is commonly the case in nonlinear deformable object simulation
and not a specific limitation of our method. 
That said, we proposed the ``mixing and matching'' method of using Neo-Hookean
volume preservation with other ``standard'' volumetric materials, which
ensures that the resulting material obeys the Poisson's effect for any $\nu$ for small deformations,
and still preserves volume under large deformations.
In the future, we would like to investigate more general nonlinearity ``filtering''
functions other than the ones used in Equation~\ref{eq:nonlinearity} to
achieve higher-order nonlinearity effects.

% DO NOT INCLUDE ACKNOWLEDGMENTS IN AN ANONYMOUS SUBMISSION TO ACM SIGGRAPH 
\begin{acks}
This research was sponsored in part by NSF (IIS-1911224), 
USC Annenberg Fellowship to Huanyu Chen, Bosch Research and Adobe Research.
\end{acks}

% Bibliography
\bibliographystyle{ACM-Reference-Format}
\bibliography{materialFiltering}

%%%%%%%%%%%%%%%%%%%%%%%%%%%%%%%%%%%%%%%%%%%%%%%%%%%%%%%%%%%%%%%%%

\end{document}